\documentclass{llncs}
\usepackage{amsmath}
\usepackage{amsfonts}
\usepackage{amssymb}

\newtheorem{observation}[theorem]{Observation}

\begin{document}

\title{Descriptive complexity of graph spectra}
\titlerunning{Descriptive complexity of graph spectra}  
%
\author{Anuj Dawar\inst{1} Simone Severini\inst{2}
Octavio Zapata\inst{2}}
\authorrunning{Anuj Dawar et al.} 
%
\tocauthor{Octavio Zapata}
\institute{University of Cambridge Computer Laboratory, UK\\
\and
Department of Computer Science, University College London, UK\thanks{We thank Aida Abiad, Chris Godsil, Robin Hirsch and David Roberson for fruitful discussions. This work was supported by CONACyT, EPSRC and The Royal Society.}}

\maketitle              

\begin{abstract}
Two graphs are co-spectral if their respective adjacency matrices
have the same multi-set of eigenvalues.  A graph is said to be
determined by its spectrum if all graphs that are co-spectral with
it are isomorphic to it.  We consider these properties in relation
to logical definability.  We show that any pair of graphs that are
elementarily equivalent with respect to the three-variable counting
first-order logic $C^3$ are co-spectral, and this is not the case
with $C^2$, nor with any number of variables if we exclude counting
quantifiers.  We also show that the class of graphs that are
determined by their spectra is definable in partial fixed-point
logic with counting.  We relate these properties to other algebraic and 
combinatorial problems.

\end{abstract}

\renewcommand{\sp}{\mathrm{sp}}
\newcommand{\C}{\mathcal{C}}
\newcommand{\DS}{\text{DS}}
\newcommand{\PFPC}{\textsc{pfpc}}
\newcommand{\LFPC}{\textsc{fpc}}
\newcommand{\PSpace}{\textsc{PSpace}}

\section{Introduction}\label{sec:intro}
The spectrum of a graph $G$ is the multi-set of eigenvalues of its adjacency matrix.  Even though it is defined in terms of the adjacency matrix of $G$, the spectrum does not, in fact, depend on the order in which the vertices of $G$ are listed.  In other words, isomorphic graphs have the same spectrum.  The converse is false: two graphs may have the same spectrum without being isomorphic.  Say that two graphs are co-spectral if they have the same spectrum.   Our aim in this paper is to study the relationship of this equivalence relation on graphs in relation to a number of other approximations of isomorphism coming from logic, combinatorics and algebra.  We also investigate the definability of co-spectrality and related notions in logic.

Specifically, we show that for any graph $G$, we can construct a formula $\phi_G$ of first-order logic with counting, using only three variables (i.e.\ the logic $C^3$) so that $H \models \phi_G$ only if $H$ is co-spectral with $G$.  From this, it follows that elementary equivalence in $C^3$ refines co-spectrality, a result that also follows from~\cite{Alzaga10}.  In contrast, we show that co-spectrality is incomparable with elementary equivalence in $C^2$, or with elementary equivalence in $L^k$ (first-order logic with $k$ variables but without counting quantifiers) for any $k$.  We show that on strongly regular graphs, co-spectrality exactly co-incides with $C^3$-equivalence.  

For definability results, we show that co-spectrality of a pair of
graphs is definable in $\LFPC$, inflationary fixed-point logic with
counting.  We also consider the property of a graph $G$ to be
\emph{determined by its spectrum}, meaning that all graphs co-spectral
with $G$ are isomorphic with $G$.  We establish that this property is
definable in \emph{partial fixed-point logic with counting} ($\PFPC$).

In section \ref{pre}, we construct some basic first-order formulas
that we use to prove various results later, and we also review some
well-known facts in the study of graph spectra. In section \ref{wlk},
we make explicit the connection between the spectrum of a graph and
the total number of closed walks on it. Then we discuss aspects of the
class of graphs that are uniquely determined by their spectra, and
establish that co-spectrality on the class of all graphs is refined by $C^3$-equivalence. Also, we show a lower bound for the distinguishability of graph spectra in the finite-variable logic. In section \ref{sym}, we give an overview of a combinatorial algorithm (named after Weisfeiler and Leman) for distinguishing between non-isomorphic graphs, and study the relationship with other algorithms of algebraic and combinatorial nature. Finally, in section \ref{pfp}, we establish some results about the logical definability of co-spectrality and of the property of being a graph determined by its spectrum.

\section{Preliminaries}\label{pre}

Consider a first-order language $L = \{E\}$, where $E$ is a binary relation symbol interpreted as an irreflexive symmetric binary relation called \emph{adjacency}. Then an $L$-structure $G=(V_G,E_G)$ is called a \emph{simple undirected graph}. The domain $V_G$ of $G$ is called the \emph{vertex set} and its elements are called \emph{vertices}. The unordered pairs of vertices in the interpretation $E_G$ of $E$ are called \emph{edges}. Formally, a \emph{graph} is an element of the elementary class axiomatised by the first-order $L$-sentence: $\forall x\forall y(\lnot E(x,x)\land(E(x,y)\rightarrow E(y,x))). $

The \emph{adjacency matrix} of an $n$-vertex graph $G$ with vertices
$v_1,\ldots,v_n$ is the $n\times n$ matrix $A_G$ with $(A_G)_{ij}=1$
if vertex $v_i$ is adjacent to vertex $v_j$, and $(A_G)_{ij}=0$
otherwise. By definition, every adjacency matrix is real and symmetric
with diagonal elements all equal to zero. A \emph{permutation matrix}
$P$ is a binary matrix with a unique $1$ in each row and column.
Permutation matrices are orthogonal matrices so the inverse $P^{-1}$
of $P$ is equal to its transpose $P^{T}$. Two graphs $G$ and $H$ are
\emph{isomorphic} if there is a bijection $h$ from $V_G$ to $V_H$ that
preserves adjacency. The existence of such a map is denoted by $G\cong H$. From this definition it is not difficult to see that two graphs $G$ and $H$ are isomorphic if, and only if, there exists a permutation matrix $P$ such that $A_GP =PA_H. $

The \emph{characteristic polynomial} of an $n$-vertex graph $G$ is a polynomial in a single variable $\lambda$ defined as $p_G(\lambda):=det(\lambda I - A_G),$ 
where $det(\cdot)$ is the operation of computing the determinant of the matrix inside the parentheses, and $I$ is the identity matrix of the same order as $A_G$. The \emph{spectrum} of $G$ is the multi-set $\sp(G):=\{\lambda:p_G(\lambda)=0\}$, where each root of $p_G(\lambda)$ is considered according to its multiplicity. If $\theta\in\sp(G)$ then $\theta I - A_G$ is not invertible, and so there exists a nonzero vector $u$ such that $A_Gu=\theta u$. A vector like $u$ is called an \emph{eigenvector} of $G$ corresponding to $\theta$. The elements in $\sp(G)$ are called the \emph{eigenvalues} of $G$. Two graphs are called \emph{co-spectral} if they have the same spectrum.

The \emph{trace} of a matrix is  the sum of all its diagonal
elements. By the definition of matrix multiplication, for any two matrices $A,B$ we have $
tr(AB)=tr(BA), $ where $tr(\cdot)$ is the operation of computing the trace of the matrix inside the parentheses. Therefore, if $G$ and $H$ are two isomorphic graphs then $
tr(A_H)=tr(P^{T}A_GP)=tr(A_GPP^{T})=tr(A_G) $ and so, $tr(A^k_G)=tr(A^k_H)$ for any $k\geq 0$. 

By the spectral decomposition theorem, computing the trace of the $k$-th powers of a real symmetric matrix $A$ will give the sum of the $k$-th powers of the eigenvalues of $A$. Assuming that $A$ is an $n\times n$ matrix with (possibly repeated) eigenvalues $\lambda_1,\dots,\lambda_n$, the \emph{elementary symmetric polynomials} $e_k$ in the eigenvalues are the sum of all distinct products of $k$ distinct eigenvalues: 
\[
\begin{array}{l}
e_0(\lambda_1,\dots,\lambda_n):=1;~~~~~~~ e_1(\lambda_1,\dots,\lambda_n):=\sum_{i=1}^n \lambda_i;\\
e_k(\lambda_1,\dots,\lambda_n):=\sum_{1\leq i_1<\cdots<i_k\leq
  n}\lambda_{i_1}\cdots\lambda_{i_k}\ \ \text{ for $1\leq k\leq n$}.
\end{array}
\]
This expressions are the coefficients of the characteristic polynomial of $A$ modulo a $1$ or $-1$ factor. That is, 
\begin{align*}
det(\lambda I-A)&=\prod_{i=1}^{n}(\lambda-\lambda_i)\\
&=\lambda^n-e_1(\lambda_1,\dots,\lambda_n)\lambda^{n-1}+\dots+(-1)^{n}e_n(\lambda_1,\dots,\lambda_n)\\
&=\sum_{k=0}^{n}(-1)^{n+k}e_{n-k}(\lambda_1,\dots,\lambda_n)\lambda^k.
\end{align*}

So if we know $s_k(\lambda_1,\dots,\lambda_n):=\sum_{i=1}^{n}\lambda_i^k$ for $k=1,\dots,n$, then using \emph{Newton's identities}: $$ e_k(\lambda_1,\dots,\lambda_n)=\frac{1}{k}\sum_{j=1}^{k}(-1)^{j-1}e_{k-j}(\lambda_1,\dots,\lambda_n)s_k(\lambda_1,\dots,\lambda_n)\ \ \text{for $1\leq k\leq n$}, $$ we can obtain all the symmetric polynomials in the eigenvalues, and so we can reconstruct the characteristic polynomial of $A$.\\ 

\begin{proposition}~\label{prop:pre1}
For $n$-vertex graphs $G$ and $H$, the following are equivalent:
\begin{itemize}
\item[$(i)$] $G$ and $H$ are co-spectral;
\item[$(ii)$] $G$ and $H$ have the same characteristic polynomial;
\item[$(iii)$] $tr(A_G^k) = tr(A_H^k)$ for $1\leq k \leq n$.
\end{itemize}
\end{proposition}

\newcommand{\GF}{\mathrm{GF}}
\newcommand{\1}{\mathbf{1}}

\section{Spectra and Walks}~\label{wlk}
Given a graph $G$, a \emph{walk of length} $l$ in $G$ is a sequence $(v_0,v_1,\dots,v_{l})$ of vertices of $G$, such that consecutive vertices are adjacent in $G$. Formally, $(v_0,v_1,\dots,v_{l})$ is a walk of length $l$ in $G$ if, and only if, $\{v_{i-1},v_i\}\in E_G$ for $1\leq i\leq l$. We say that the walk $(v_0,v_1,\dots,v_{l})$ \emph{starts} at $v_0$ and \emph{ends} at $v_l$. A walk of length $l$ is said to be \emph{closed} (or $l$-\emph{closed}, for short) if it starts and ends in the same vertex.

Since the $ij$-th entry of $A_G^l$  is precisely the number of walks of length $l$ in $G$ starting at $v_i$ and ending at $v_j$, by Proposition~\ref{prop:pre1}, we have that the spectrum of $G$ is completely determined if we know the total number of closed walks for each length up to the number of vertices in $G$. Thus, two graphs $G$ and $H$ are co-spectral if, and only if, the total number of $l$-closed walks in $G$ is equal to the total number of $l$-closed walks in $H$ for all $l\geq 0$.

For an example of co-spectral non-isomorphic graphs, let $G=C_4\cup
K_1$ and $H=K_{1,4}$, where $C_n$ is the $n$-vertex cycle, $K_n$ is the complete $n$-vertex graph,
$K_{n,m}$ the complete $(n+m)$-vertex bipartite graph, and ``$\cup$''
denotes the disjoint union of two graphs. The spectrum of both $G$ and
$H$ is the multi-set $\{-2,0,0,0,2\}$. 
However, $G$ contains an isolated vertex while $H$ is a connected graph.

\subsection{Finite Variable Logics with Counting}
For each positive integer $k$, let $C^k$ denote the fragment of first-order logic in which only $k$ distinct variables can be used but we allow \emph{counting quantifiers}: so for each $i\geq 1$ we have a quantifier $\exists^i$ whose semantics is defined so that $\exists^i x \phi$ is true in a structure if there are at least $i$ distinct elements that can be substituted for $x$ to make $\phi$ true. We use the abbreviation $\exists^{=i} x \phi$ for the formula $\exists^i x \phi \land\lnot \exists^{i+1} x \phi$ that asserts the existence of exactly $i$ elements satisfying $\phi$. We write $G \equiv_C^k H$ to denote that the graphs $G$ and $H$ are not distinguished by any formula of $C^k$. Note that \emph{$C^k$-equivalence} is the usual first-order elementary equivalence relation restricted to formulas using at most $k$ distinct variables and possibly using counting quantifiers.

We show that for integers $k,l$, with $k\geq 0$ and $l\geq 1$, there is a formula $\psi^l_k(x,y)$ of $C^3$ so that for any graph $G$ and vertices $v,u\in V_G$, $G\models\psi^{l}_k[v,u]$ if, and only if, there are exactly $k$ walks of length $l$ in $G$ that start at $v$ and end at $u$.  We define this formula by induction on $l$.  Note that in the inductive definition, we refer to a formula $\psi^l_k(z,y)$.  This is to be read as the formula $\psi^l_k(x,y)$ with all occurrences of $x$ and $z$ (free or bound) interchanged. In particular, the free variables of $\psi^l_k(x,y)$ are exactly $x,y$ and those of $\psi^l_k(z,y)$ are exactly $z,y$.

For $l = 1$, the formulas are defined as follows: $$\psi^1_0(x,y):= \lnot E(x,y);\ \ \psi^1_1(x,y):= E(x,y); $$ $$ \mathrm{~~~and}\ \ \psi^1_k(x,y):=\mathrm{false\ \ \ \ for}\ \ k>1.$$ 

For the inductive case, we first introduce some notation. Say that a collection $(i_1,k_1),\dots,(i_r,k_r)$ of pairs of integers, with $i_j\geq 1$ and $k_j\geq 0$ is an \emph{indexed partition} of $k$ if the $k_1,\dots,k_r$ are pairwise distinct and $k =\sum_{j=1}^{r} i_j k_j$. That is, we partitioned $k$ into $\sum_{j=1}^r i_j$ distinct parts, and there are exactly $i_j$ parts of size $k_j$ where $j=1,\dots,r$. Let $K$ denote the set of all indexed partitions of $k$ and note that this is a finite set. 

Now, assume we have defined the formulas $\psi^l_k(x,y)$ for all values of $k\geq 0$. We proceed to define them for $l+1$ $$\psi^{l+1}_0(x,y):= \forall z(E(x,z)\rightarrow\psi^{l}_0(z,y))$$
$$\psi^{l+1}_k(x,y):= \bigvee_{(i_1,k_1),\dots,(i_r,k_r)\in K}\Big(\big(\bigwedge_{j=1}^r\exists^{=i_j}z\ (E(x,z) \land \psi^{l}_{k_j}(z,y)\big)\land\exists^{=d} z\ E(x,z)\Big),$$ where $d=\sum_{j=1}^r i_j$. Note that without allowing counting quantification it would be necessary to use many more distinct variables to rewrite the last formula.

Given an $n$-vertex graph $G$, as noted before $(A_G^{l})_{ij}$ is equal to the number of walks of length $l$ in $G$ from vertex $v_i$ to vertex $v_j$, so $(A_G^{l})_{ij}=k$ if, and only if, $G\models \psi^{l}_k (v_i,v_j)$. Once again, let $K$ denote the set of indexed partitions of $k$. For each integer $k\geq 0$ and $l\geq 0$, we define the sentence 
$$\phi^{l}_{k} := \bigvee_{(i_1,k_1),\dots,(i_r,k_r)\in K} \Big(\bigwedge_{j=1}^{r} \exists^{=i_j}x\exists y\big( x=y\land \psi^{l}_{k}(x,y)\big)\Big).$$
Then we have $G\models \phi^{l}_k$ if, and only if, the total number of closed walks of length $l$ in $G$ is exactly $k$. Hence $G\models \phi^{l}_k$ if, and only if, $tr(A_G^{l})=k$. Thus, we have the following proposition.
\begin{proposition}\label{prop:wlk1}
If $G\equiv_C^3 H$ then $G$ and $H$ are co-spectral.
\end{proposition}
\begin{proof}
Suppose $G$ and $H$ are two non-cospectral graphs. Then there is some $l$ such that $tr(A_G^{l})\neq tr(A_H^{l})$, i.e. the total number of closed walks of length $l$ in $G$ is different from the total number of closed walks of length $l$ in $H$ (see Proposition \ref{prop:pre1}). If $k$ is the total number of closed walks of length $l$ in $G$, then $G\models\phi^{l}_k$ and $H\not\models\phi^{l}_k$. Since $\phi^{l}_k$ is a sentence of $C^3$, we conclude that $G\not\equiv_C^3 H$. \qed
\end{proof}

For any graph $G$ and $l\geq 1$, there exists a positive integer $k_{l}$ such that $tr(A_G^{l})=k_{l}$. Since having the traces of powers of the adjacency matrix of $G$ up to the number of vertices is equivalent to having the spectrum of $G$, we can define a sentence $$\phi_G := \bigwedge_{l=1}^n \phi^{l}_{k_l}$$ of $C^3$ such that for any graph $H$, we have $H\models\phi_G$ if, and only if, $\sp(G)=\sp(H)$. 

\subsection{Graphs Determined by Their Spectra}

We say that a graph $G$ is \emph{determined by its spectrum} (for short, DS) when for any graph $H$, if $\sp(G)=\sp(H)$ then $G\cong H$. In words, a graph is determined by its spectrum when it is the only graph up to isomorphism with a certain spectrum. In Proposition \ref{prop:wlk1} we saw that $C^3$-equivalent graphs are necessarily co-spectral. That is, if two graphs $G$ and $H$ are $C^3$-equivalent then $G$ and $H$ must have the same spectrum. It thus follows that being identified by $C^3$ is weaker than being determined by the spectrum, so there are more graphs identified by $C^3$ than graphs determined by their spectra. 

\begin{observation}
On the class of all finite graphs, $C^3$-equivalence refines co-spectrality.
\end{observation}

In general, determine whether a graph has the DS property (i.e., the equivalence class induced by having the same spectrum coincides with its isomorphism class) is an open problem in spectral graph theory (see, e.g. \cite{Van03}). Given a graph $G$ and a positive integer $k$, we say that the logic $C^k$ \emph{identifies} $G$ when for all graphs $H$, if $G\equiv_C^k H$ then $G\cong H$. Let $\C^k_n$ be the class of all $n$-vertex graphs that are identified by $C^k$. 
Since $C^2$-equivalence corresponds to indistinguishability by the $1$-dimensional Weisfeiler-Lehman algorithm~\cite{Immerman90}, from a classical result of Babai, Erd\H{o}s and Selkow~\cite{Babai80}, it follows that $\C^{2}_n$ contains almost all $n$-vertex graphs. Let $\DS_n$ be the class of all DS $n$-vertex graphs.

The $1$-dimensional Weisfeiler-Lehman algorithm (see Section
\ref{sym}) does not distinguish any pair of non-isomorphic regular graphs of the same degree with the same number of vertices. Hence, if a regular graph is not determined up to isomorphism
by its number of vertices and its degree, then it is not in $\C^2_n$.
However, there are regular graphs that are determined by their number of
vertices and their degree. For instance, the complete graph on $n$
vertices, which gives an example of a graph in  $\DS_n\cap \C^2_n$.

Let $T$ be a tree on $n$ vertices.  By a well-known result from Schwenk~\cite{Schwenk73}, with probability one there exists another tree $T^{\prime}$ such that $T$ and $T^{\prime}$ are co-spectral but not isomorphic. From a result of Immerman and Lander~\cite{Immerman90} we know that all trees are identified by $C^2$. Hence $T$ is an example of a graph in $\C^2_n$ and not in $\DS$. On the other hand, the disjoint union of two complete graphs with the same number of vertices is a graph which is determined by its spectrum. That is, $2K_m$ is DS (see~\cite[Section~6.1]{Van03}). For each $m>2$ it is possible to construct a connected regular graph $G_{2m}$ with the same number of vertices and the same degree as $2K_m$. Hence $G_{2m}$ and $2K_m$ are not distinguishable in $C^2$ and clearly not isomorphic. This shows that co-spectrality and elementary equivalence with respect to the two-variable counting logic is incomparable.

From a result of Babai and Ku\v{c}era~\cite{Babai79}, we know that a graph randomly selected from the uniform distribution over the class of all unlabeled $n$-vertex graphs (which has size equal to $2^{n(n-1)/2}$) is not identified by $C^2$ with probability equal to $(o(1))^n$. Moreover, in~\cite{Kucera87} Ku\v{c}era presented an efficient algorithm for labelling the vertices of random regular graphs from which it follows that the fraction of regular graphs which are not identified by $C^3$ tends to $0$ as the number of vertices tends to infinity. Therefore, almost all regular $n$-vertex graphs are in $\C^3_n$. Summarising, $\DS_n$ and $\C_n^2$ overlap and both are contained in $\C_n^3$.

\subsection{Lower Bounds}

Having established that $C^3$-equivalence is a refinement of co-spectrality, we now look at the relationship of the latter with equivalence in finite variable logics without counting quantifiers.  First of all, we note that some  co-spectral graphs can be distinguished by a formula using just two variables and no counting quantifiers. 

\begin{proposition}\label{prop:lb1}
There exists a pair of co-spectral graphs that can be distinguished in first-order logic with only two variables.
\end{proposition}
\begin{proof}

Let us consider the following two-variable first-order sentence: 
$$\psi:= \exists x\forall y\ \lnot E(x, y). $$

For any graph $G$ we have that $G\models\psi$ if, and only if, there is an isolated vertex in $G$. Hence $C_4\cup K_1\models\psi$ and $K_{1,4}\not\models\psi$. Therefore, $C_4\cup K_1\not\equiv^2 K_{1,4}$. \qed
\end{proof}

Next, we show that counting quantifiers are essential to the argument
from the previous section in that co-spectrality is not subsumed by
equivalence in any finite-variable fragment of first-order logic in
the absence of such quantifiers.  Let $L^k$ denote the fragment of
first-order logic in which each formula has at most $k$ distinct
variables. 


For each $r,s\geq0$, the \emph{extension axiom} $\eta_{r,s}$ is the first-order sentence $$
\forall x_1\dots\forall x_{r+s}\Bigg(\bigg(\bigwedge_{i\neq j} x_i\neq x_j\bigg)\rightarrow\exists y\bigg(\bigwedge_{i\leq r} E(x_i,y)\land\bigwedge_{i>r}\lnot E(x_i,y)\wedge x_i\neq y\bigg)\Bigg). $$

A graph $G$ satisfies the \emph{$k$-extension property} if $G\models\eta_{r,s}$ and $r+s=k$. In~\cite{Kolaitis92} Kolaitis and Vardi proved that if the graphs $G$ and $H$ both satisfy the $k$-extension property, then there is no formula of $L^{k}$ that can distinguish them. If this happens, we write $G\equiv^{k} H$. Fagin~\cite{Fagin76} proved that for each $k\geq 0$, almost all graphs satisfy the $k$-extension property. Hence almost all graphs are not distinguished by any formula of $L^k$.

Let $q$ be a prime power such that $q\equiv 1$ (mod 4). The \emph{Paley graph} of order $q$ is the graph $P(q)$ with vertex set $\GF(q)$, the finite field of order $q$, where two vertices $i$ and $j$ are adjacent if there is a positive integer $x$ such that $x^2 \equiv (i-j)$ (mod $q$). Since $q\equiv 1$ (mod 4) if, and only if, $x^2\equiv -1$ (mod $q$) is solvable, we have that $-1$ is a square in $\GF(q)$ and so, $(j-i)$ is a square if and only if $-(i-j)$ is a square. Therefore, adjacency in a Paley graph is a symmetric relation and so, $P(q)$ is undirected. Blass, Exoo and Harary~\cite{Blass81} proved that  if $q$ is greater than $k^2 2^{4k}$, then $P(q)$ satisfies the $k$-extension property.

Now, let $q=p^r$ with $p$ an odd prime, $r$ a positive integer, and $q\equiv 1$ (mod $3$). The \emph{cubic Paley graph} $P^3(q)$ is the graph whose vertices are elements of the finite field $\GF(q)$, where two vertices $i,j\in\GF(q)$ are adjacent if and only if their difference is a cubic residue, i.e. $i$ is adjacent to $j$ if, and only if, $i-j = x^3$ for some $x\in\GF(q)$. Note that $-1$ is a cube in $\GF(q)$ because $q\equiv 1$ (mod $3$) is a prime power, so $i$ is adjacent to $j$ if, and only if, $j$ is adjacent to $i$. In~\cite{Ananchuen06} it has been proved that $P^3(q)$ has the $k$-extension property whenever $q\geq k^2 2^{4k-2}$.  

The \emph{degree} of vertex $v$ in a graph $G$ is the number $d(v):= |\{\{v,u\}\in E:u\in V_G\}|$ of vertices that are adjacent to $v$. A graph $G$ is \emph{regular} of degree $d$ if every vertex is adjacent to exactly $d$ other vertices, i.e.\ $d(v)=d$ for all $v\in V_G$. So, $G$ is regular of degree $d$ if, and only if, each row of its adjacency matrix adds up to $d$. It can been shown that the Paley graph $P(q)$ is regular of degree $(q-1)/2$~\cite{GodsilRoyle}. Moreover, it has been proved that the cubic Paley graph $P^3(q)$ is regular of degree $(q-1)/3$~\cite{Elsawy12}.  

\begin{lemma}\label{lem:lb1}
Let $G$ be a regular graph of degree $d$. Then $d\in\sp(G)$ and for each $\theta\in\sp(G)$, we have $|\theta|\leq d$. Here $|\cdot|$ is the operation of taking the absolute value.
\end{lemma}
\begin{proof}
Let us denote by $\1$ the all-ones vector. Then $A_G\1=d\1$. Therefore, $d\in\sp(G)$. Now, let $s$ be such that $|s| > d$. Then, for each row $i$, $$|S_{ii}|>\sum_{j\neq i}|S_{ij}|$$ where $S=s I - A_G$. Therefore, the matrix $S$ is strictly diagonally dominant, and so $det(sI-A_G)\neq 0$. Hence $s$ is not an eigenvalue of $G$. \qed
\end{proof}

\begin{lemma}\label{lem:lb2}
Let $G$ and $H$ be regular graphs of distinct degrees. Then $G$ and $H$ do not have the same spectrum.
\end{lemma}
\begin{proof}
Suppose that $G$ is regular of degree $s$ and $H$ is regular of degree $t$, with $s\neq t$. Then $A_G\1=s\1$ and $A_H\1=t\1$, where $\1$ is the all-ones vector. Therefore, $s$ is the greatest eigenvalue in the spectrum of $G$ and $t$ is the greatest eigenvalue in the spectrum of $H$. Hence $\sp(G)\neq\sp(H)$. \qed
\end{proof}

\begin{proposition}\label{prop:lb2}
For each $k\geq 1$, there exists a pair $G_k, H_k$ of graphs which are not co-spectral, such that $G_k$ and $H_k$ are not distinguished by any formula of $L^{k}$.
\end{proposition}
\begin{proof}
For any positive integer $r$ we have that $13^r\equiv 1$ (mod $3$) and $13^r\equiv 1$ (mod $4$). For each  $k\geq 1$, let $r_k$ be the smallest integer greater than $2 (k \log(4)+\log(k))/\log(13)$, and let $q_k = 13^{r_k}$. Hence $q_k>k^2 2^{4k}$. Now, let  $G_k=P(q_k)$ and $H_k=P^3(q_k)$. Then $G_k$ and $H_k$ both satisfy the $k$-extension property, and so $G_k \equiv^{k} H_k$. Since the degree of $G_k$ is $(13^{r_k}-1)/2$ and the degree of $H_k$ is $(13^{r_k}-1)/3$, by Lemma~\ref{lem:lb1} we conclude that $\sp(G_k)\neq\sp(H_k)$. \qed
\end{proof}

So having the same spectrum is a property of graphs that does not follows from any finite collection of extension axioms, or equivalently, from any first-order sentence with asymptotic probability 1.

\renewcommand{\b}[1]{\mathbf{#1}}
\newcommand{\aut}{\mathrm{Aut}}
\newcommand{\srg}{\mathrm{srg}}
\newcommand{\tp}{\mathrm{tp}}

\section{Isomorphism Approximations}\label{sym}

\subsection{WL Equivalence}

The automorphism group $\aut(G)$ of $G$ acts naturally on the set $V^k_G$ of all $k$-tuples of vertices of $G$, and the set of orbits of $k$-tuples under the action of $\aut(G)$ form a corresponding partition of $V^k_G$. The \emph{$k$-dimensional Weisfeiler-Leman algorithm} is a combinatorial method that tries to approximate the partition induced by the orbits of $\aut(G)$ by labelling the $k$-tuples of vertices of $G$. For the sake of completeness, here we give a brief overview of the algorithm.

The $1$-dimensional Weisfeiler-Leman algorithm has the following steps: first, label each vertex $v\in V_G$ by its degree $d(v)$. The set $N(v):=\{u:\{v,u\}\in E_G\}$ is called the \emph{neighborhood} of $v\in V_G$ and so, the degree of $v$ is just the number of neighbours it has, i.e. $d(v)=|N(v)|$. In this way we have defined a partition $P_{0}(G)$ of $V_G$. The number of labels is equal to the number of different degrees. Hence $P_{0}(G)$ is the degree sequence of $G$. Then, relabel each vertex $v$ with the multi-set of labels of its neighbours, so each label $d(v)$ is substituted for $\{d(v), \{d(u):u\in N(v)\}\}$. Since these are multi-sets they might contain repeated elements. We get then a partition $P_{1}(G)$ of $V_G$ which is either a refinement of $P_{0}(G)$ or identical to $P_{0}(G)$. Inductively, the partition $P_{t}(G)$ is obtained from the partition $P_{t-1}(G)$, by constructing for each vertex $v$ a new multi-set that includes the labels of its neighbours, as it is done in the previous step. The algorithm halts as soon as the number of labels does not increase anymore. We denote the resulting partition of $V_G$ by $P^1_G$.  

Now we describe the algorithm for higher dimensions. Recall that we are working in the first-order language of graphs $L=\{E\}$. Now, for each graph $G$ and each $k$-tuple $\b{v}$ of vertices of $G$ we define the \emph{(atomic) type of $\b{v}$} in $G$ as the set $\tp^k_G(\b{v})$ of all atomic $L$-formulas $\phi(\b{x})$ that are true in $G$ when the variables of $\b{x}$ are substituted for vertices of $\b{v}$. More formally, for $k>1$ we let
\[
\tp^k_G(\b{v}) := \{\phi(\b{x}): |\b{x}|\leq k, G\models\phi(\b{v})\}
\] 
where, $|\b{x}|$ denotes the number of entries the tuple $\b{x}$ have, and each $\phi(\b{x})$ is either $x_i=x_j$ or $E(x_i,x_j)$ for $1\leq i, j\leq k$. Essentially, the formulas of $\tp^k_G(\b{v})$ give us the complete information about the structural relations that hold  between the vertices of $\b{v}$. If $u\in V_G$ and $1\leq i\leq k$, let $\b{v}_i^u$ denote the result of substituting $u$ in the $i$-th entry of $\b{v}$. 

For each $k> 1$ the $k$-dimensional Weisfeiler-Leman algorithm proceeds as follows: first, label the $k$-tuples of vertices with their types in $G$, so each $k$-tuple $\b{v}$ is labeled with $\ell_0(\b{v}):=\tp^k_G(\b{v})$; this induces a partition $P^k_0(G)$ of the $k$-tuples of vertices of $G$. Inductively, refine the partition $P^k_i(G)$ of $V^k_G$ by relabelling the $k$-tuples so that each label $\ell_i(\b{v})$ is substituted for $\ell_{i+1}(\b{v}):=\{\ell_{i}(\b{v}),\{\ell_{i}(\b{v}_1^u),\dots,\ell_{i}(\b{v}_k^u):u\in V_G\}\}$. The algorithm continues refining the partition of $V^k_G$ until it gets to a step $t\geq 1$, where $P_t^k(G) = P_{t-1}^k(G)$; then it halts. We denote the resulting partition of $V^k_G$ by $P^k_G$. 

Notice that for any fixed $k\geq 1$, the partition $P^k_G$ of $k$-subsets is obtained after at most $|V_G|^k$ steps. If the partitions $P^k_G$ and $P^k_H$ of graphs $G$ and $H$ are the same multi-set of labels obtained by the $k$-dimensional Weisfeiler-Leman algorithm, we say that $G$ and $H$ are \emph{$k$-WL equivalent}. In~\cite{Cai92}, Cai, F\"urer and Immerman proved that two graphs $G$ and $H$ are $C^{k+1}$-equivalent if, and only if, $G$ and $H$ are $k$-WL equivalent. 

\subsection{Symmetric Powers}
The \emph{$k$-th symmetric power} $G^{\{k\}}$ of a graph $G$ is a graph where each vertex represents a $k$-subset of vertices of $G$, and two $k$-subsets are adjacent if their symmetric difference is an edge of $G$. Formally, the vertex set $V_{G^{\{k\}}}$ of $G^{\{k\}}$ is defined to be the set of all subsets of $V_G$ with exactly $k$ elements, and for every pair of $k$-subsets of vertices $V=\{v_1,\dots,v_k\}$ and $U=\{u_1,\dots,u_k\}$, we have $\{V, U\}\in E_{G^{\{k\}}}$ if, and only if, $(V\smallsetminus U)\cup(U\smallsetminus V) \in E_G$. The symmetric powers are related to a natural generalisation of the concept of a walk in a graph. A \emph{$k$-walk} of length $l$ in $G$ is a sequence $(V_0, V_1,\dots, V_l)$ of $k$-subsets of vertices, such that the symmetric difference of $V_{i-1}$ and $V_{i}$ is an edge of $G$ for $1\leq i\leq l$. A $k$-walk is said to be closed if $V_0=V_l$. The connection with the symmetric powers is that a $k$-walk in $G$ corresponds to an ordinary walk in $G^{\{k\}}$. Therefore, two graphs have the same total number of closed $k$-walks of every length if, and only if, their $k$-th symmetric powers are co-spectral. For each $k\geq 1$, there exist infinitely many pairs of non-isomorphic graphs $G$ and $H$ such that the $k$-th symmetric powers $G^{\{k\}}$ and $H^{\{k\}}$ are co-spectral~\cite{Barghi09}.

Alzaga, Iglesias and Pignol~\cite{Alzaga10} have shown that given two graphs $G$ and $H$, if $G$ and $H$ are $2k$-WL equivalent, then their $k$-th symmetric powers $G^{\{k\}}$ and $H^{\{k\}}$ are co-spectral. This two facts combined allow us to deduce the following generalisation of Proposition~\ref{prop:wlk1}.

\begin{proposition}\label{prop:sym1}
Given graphs $G$ and $H$ and a positive integer $k$, if $G\equiv_{C}^{2k+1}H$ then $G^{\{k\}}$ and $H^{\{k\}}$ are co-spectral.
\end{proposition}

\subsection{Cellular Algebras}

Originally, Weisfeiler and Leman~\cite{WL68} presented their algorithm in terms of algebras of complex matrices. Given two matrices $A$ and $B$ of the same order, their \emph{Schur product} $A\circ B$ is defined by $(A\circ B)_{ij}:=A_{ij}B_{ij}$. For a complex matrix $A$, let $A^{\ast}$ denote the adjoint (or conjugate-transpose) of $A$. A \emph{cellular algebra} $W$ is an algebra of square complex matrices that contains the identity matrix $I$, the all-ones matrix $J$, and is closed under adjoints and Schur multiplication. Thus, every cellular algebra has a unique basis $\{A_1,\dots,A_m\}$ of binary matrices which is closed under adjoints and such that $\sum_i A_i = J$.

The smallest cellular algebra is the one generated by the span of $I$ and $J$. The cellular algebra of an $n$-vertex graph $G$ is the smallest cellular algebra $W_G$ that contains $A_G$. Two cellular algebras $W$ and $W^{\prime}$ are isomorphic if there is an algebra isomorphism $h:W\to W^{\prime}$, such that $h(A\circ B)=h(A)\circ h(B)$, $h(A)^{\ast}=h(A^{\ast})$ and $h(J)=J$. Given an isomorphism $h:W\to W^{\prime}$ of cellular algebras, for all $A\in W$ we have that $A$ and $h(A)$ are co-spectral (see Lemma~3.4 in~\cite{Friedland89}). So the next result is immediate.

\begin{proposition}\label{prop:ca1}
Two graphs $G$ and $H$ are co-spectral if there is an isomorphism of $W_G$ and $W_H$ that maps $A_G$ to $A_H$.
\end{proposition}

In general, the converse of Proposition~\ref{prop:ca1} is not true. That is, there are known pairs of co-spectral graphs whose corresponding cellular algebras are non-isomorphic (see, e.g. \cite{Barghi09}).
The elements of the standard basis of a cellular algebra correspond to the ``adjacency matrices'' of a corresponding coherent configuration. Coherent configurations where introduced by Higman in~\cite{Higman75} to study finite permutation groups. 
Coherent configurations are stable under the $2$-dimensional Weisfeiler-Leman algorithm. Hence two graphs $G$ and $H$ are $2$-WL equivalent if, and only if, there is an isomorphism of $W_G$ and $W_H$ that maps $A_G$ to $A_H$. 

\begin{proposition}\label{obs:ca1}
Given graphs $G$ and $H$ with cellular algebras $W_G$ and $W_H$, $G\equiv_C^3 H$ if, and only if, there is an isomorphism of $W_G$ and $W_H$ that maps $A_G$ to $A_H$. 
\end{proposition}

\subsection{Strongly Regular Graphs}

A \emph{strongly regular graph} $\srg(n,r,\lambda,\mu)$ is a regular $n$-vertex graph of degree $r$ such that each pair of adjacent vertices has $\lambda$ common neighbours, and each pair of nonadjacent vertices has $\mu$ common neighbours. The numbers $n,r,\lambda,\mu$ are called the \emph{parameters} of $\srg(n,r,\lambda,\mu)$. It can be shown that the spectrum of a strongly regular graph is determined by its parameters~\cite{GodsilRoyle}. The complement of a strongly regular graph is strongly regular. Moreover, co-spectral strongly regular graphs have co-spectral complements. That is, two strongly regular graphs having the same parameters are co-spectral. Recall $J$ is the all-ones matrix.

\begin{lemma}
If $G$ is a strongly regular graph then $\{I , A_G, (J - I - A_G)\}$ form the basis for its corresponding cellular algebra $W_G$. 
\end{lemma}
\begin{proof}
By definition, $W_G$ has a unique basis $\mathcal{A}$ of binary matrices closed under adjoints and so that \[\sum_{A\in\mathcal{A}} A = J.\]
Notice that $I,A_G$ and $J - I - A_G$ are binary matrices such that $I^{\ast}=I$, $A_G^{\ast}=A_G$ and $
(J - I - A_G)^{\ast} = J - I - A_G.$ Furthermore, 
\[I + A_G + (J - I - A_G) = J.\]
\end{proof}

There are known pairs of non-isomorphic strongly regular graphs with the same parameters (see, e.g. \cite{Brouwer84}). These graphs are not distinguished by the $2$-dimensional Weisfeiler-Leman algorithm 
since there is an algebra isomorphism that maps the adjacency matrix of one to the adjacency matrix of the other. Thus, for strongly regular graphs the converse of Proposition~\ref{prop:ca1} holds.

\begin{lemma}\label{lem:ca1} 
If $G$ and $H$ are two co-spectral strongly regular graphs, then there exists an isomorphism of $W_G$ and
$W_H$ that maps $A_G$ to $A_H$.
\end{lemma}
\begin{proof}
The cellular algebras $W_G$ and $W_H$ of $G$ and $H$ have standard basis $\{I , A_G, (J - I - A_G)\}$ and $\{I , A_H, (J - I - A_H)\}$, respectively. Since $G$ and $H$ are co-spectral, there exist an orthogonal matrix $Q$ such that $QA_GQ^{T}=A_H$ and $Q(J - I - A_G)Q^{T}=(J - I - A_H)$. In~\cite{Friedland89}, Friedland has shown that two cellular algebras with standard bases $\{A_1,\dots,A_m\}$ and $\{B_1,\dots,B_m\}$ are isomorphic if, and only if, there is an invertible matrix $M$ such that $MA_iM^{-1}=B_i$ for $1\leq i\leq m$. As every orthogonal matrix is invertible, we can conclude that there exists an isomorphism of $W_G$ and
$W_H$ that maps $A_G$ to $A_H$.
\end{proof}


\begin{proposition}\label{prop:ca2}
Given two strongly regular graphs $G$ and $H$, the following statements are equivalent:
\begin{enumerate}
\item $G\equiv_C^3 H$;
\item $G$ and $H$ are co-spectral;
\item there is an isomorphism of $W_G$ and $W_H$ that maps $A_G$ to $A_H$.
\end{enumerate}
\end{proposition} 
\begin{proof}
Proposition~\ref{prop:wlk1} says that for all graphs $(1)$ implies $(2)$. From Proposition~\ref{obs:ca1}, we have $(1)$ if, and only if, $(3)$. By Lemma~\ref{lem:ca1}, if $(2)$ then $(3)$.
\end{proof}

\newcommand{\num}{\mathrm{num}}
\newcommand{\elem}{\mathrm{elem}}
\newcommand{\lex}{\mathrm{lex}}
\newcommand{\next}{\mathrm{next}}
\newcommand{\isom}{\mathrm{isom}}
\newcommand{\cospec}{\mathrm{cospec}}
\newcommand{\summ}{\mathrm{sum}}
\newcommand{\bR}{\ensuremath{\bar{R}}}

\newcommand{\countingTerm}[1]{\#_{#1}}
\newcommand{\tup}[1]{\vec{#1}}

\newcommand{\struct}[1]{\ensuremath{\mathbf #1}}
\newcommand{\univ}[1]{\ensuremath{\dom(\struct #1)}} 
\newcommand{\dom}{\mathrm{dom}}

\newcommand{\logicoperator}[1]{\mathbf{#1}}
\newcommand{\ifp}{\logicoperator{ifp}}
\newcommand{\pfp}{\logicoperator{pfp}}

\section{Definability in Fixed Point Logic with Counting}\label{pfp}

In this section, we consider the definability of co-spectrality and
the property DS in fixed-point logics with counting.  To be precise, we
show that co-spectrality is definable in \emph{inflationary fixed-point logic with
counting} ($\LFPC$) and  the class of graphs that are DS is definable in
\emph{partial fixed-point logic with counting} ($\PFPC$).  It follows
that both of these are also definable  in the infinitary logic with counting, with a bounded
number of variables (see~\cite[Prop.~8.4.18]{EF99}).  Note that it is
known that $\LFPC$ can express any polynomial-time decidable property
of \emph{ordered} structures and similarly $\PFPC$ can express all
polynomial-space decidable properties of ordered structures.  It is
easy to show that co-spectrality is decidable in polynomial time and
DS is in $\PSpace$.  For the latter, note that DS can easily be
expressed by a $\Pi_2$ formula of second-order logic and therefore the
problem is in the second-level of the polynomial hierarchy.  However,
in the absence of a linear order $\LFPC$ and $\PFPC$ are strictly
weaker than the complexity classes $P$ and $\PSpace$ respectively.
Indeed, there are problems in $P$ that are not even expressible in the
infinitary logic with counting.  Nonetheless, it is in this context
without order that we establishe the definability results below.

We begin with a brief definition of the logics in question, to fix the
notation we use.  For a more detailed definition, we refer the reader
to~\cite{EF99}~\cite{Lib04}. 

$\LFPC$ is an extension of inflationary fixed-point logic with the
ability to express the cardinality of definable sets.  The logic has
two sorts of first-order variables: \emph{element variables}, which
range over elements of the structure on which a formula is interpreted
in the usual way, and \emph{number variables}, which range over some
initial segment of the natural numbers. We usually write element
variables with lower-case Latin letters $x, y, \dots$ and use
lower-case Greek letters $\mu, \eta, \dots$ to denote number
variables.  In addition, we have relational variables, each of which has an arity $m$ and an associated type from $\{\mathrm{elem},\mathrm{num}\}^m$.  $\PFPC$ is similarly obtained by allowing the
\emph{partial fixed point} operator in place of the inflationary
fixed-point operator.

For a fixed signature $\tau$, the atomic formulas of $\LFPC[\tau]$ of $\PFPC[\tau]$ are all formulas of the form $\mu
= \eta$ or $\mu \le \eta$, where $\mu, \eta$ are number variables; $s
= t$ where $s,t$ are element variables or constant symbols from
$\tau$; and $R(t_1, \dots, t_m)$, where  $R$ is a relation symbol
(i.e.\ either a symbol from $\tau$ or a relational variable) of arity
$m$ and each $t_i$ is a term of the appropriate type (either $\elem$ or $\num$, as determined by the type of $R$). The set $\LFPC[\tau]$ of 
\emph{$\LFPC$ formulas} over $\tau$ is built up from the atomic formulas by 
applying an inflationary fixed-point operator $[\ifp_{R,\tup x}\phi](\tup t) $;
forming \emph{counting terms} $\countingTerm{x} \phi$, where $\phi$ is a formula
and $x$ an element variable; forming formulas of the kind $s = t$ and $s \le t$ 
where $s,t$ are number variables or counting terms; as well as the standard 
first-order operations of negation, conjunction, disjunction, universal and 
existential quantification. Collectively, we refer to element variables and 
constant symbols as \emph{element terms}, and to number variables and counting 
terms as \emph{number terms}.  The formulas of  $\PFPC[\tau]$ are
defined analogously, but we replace the fixed-point operator rule by
the partial fixed-point:  $[\pfp_{R,\tup x}\phi](\tup t)$.

For the semantics, number terms take  values in $\{0,\ldots,n\}$,
where $n$ is the size of the structure in which they are interpreted. The semantics of atomic formulas,
fixed-points and first-order operations are defined as usual (c.f.,
e.g., \cite{EF99} for details), with comparison of number terms
$\mu \le \eta$ interpreted by comparing the corresponding integers in
$\{0,\ldots,n\}$. Finally, consider a counting term of the form
$\countingTerm{x}\phi$, where $\phi$ is a formula and $x$ an element
variable. Here the intended semantics is that $\countingTerm{x}\phi$
denotes the number (i.e.\ the element of $\{0,\ldots,n\}$) of elements that
satisfy the formula $\phi$.  

Note that, since an inflationary
fixed-point is easily expressed as a partial fixed-point, every
formula of $\LFPC$ can also be expressed as a formula of $\PFPC$.  In
the construction of formulas of these logics below, we freely use
arithmetic expressions on number variables as the relations defined by
such expressions can easily be defined by formulas of $\LFPC$.

In Section~\ref{wlk} we constructed sentences $\phi^l_k$ of $C^3$
which are satisfied in a graph $G$ if, and only if, the number of
closed walks in $G$ of length $l$ is exactly $k$.  Our first aim is to
construct a single formula of $\LFPC$ that expresses this for all $l$
and $k$.  Ideally, we would have the numbers as parameters to the
formula but it should be noted that, while the length $l$ of walks we
consider is bounded by the number $n$ of vertices of $G$, the number
of closed walks of length $l$ is not bounded by any polynomial in $n$.
Indeed, it can be as large as $n^n$.  Thus, we cannot represent the
value of $k$ by a single number variable, or even a fixed-length tuple
of number variables.  Instead, we represent $k$ as a binary relation
$K$ on the number domain.  The order on the number domain induces a
lexicographical order on pairs of numbers, which is a way of encoding
numbers in the range $0,\ldots, n^2$.  Let us write $[i,j]$ to denote
the number coded by the pair $(i,j)$.  Then, a binary relation $K$ can be
used to represent a number $k$ up to $2^{n^2}$ by its binary encoding.  To
be precise, $K$ contains all pairs $(i,j)$ such that bit position
$[i,j]$ in the binary encoding of $k$ is 1.  It is easy to define
formulas of $\LFPC$ to express arithmetic operations on numbers
represented in this way.  

Thus, we aim to construct a single formula
$\phi(\lambda,\kappa_1,\kappa_2)$ of $\LFPC,$ with three free number
variables such that $G \models \phi[l,i,j]$ if, and only if, the
number of closed walks in $G$ of length $l$ is $k$ and position
$[i,j]$ in the binary expansion of $k$ is 1.  To do this, we first
define a formula $\psi(\lambda,\kappa_1,\kappa_2,x,y)$ with free
number variables $\lambda$, $\kappa_1$ and $\kappa_2$ and free element
variables $x$ and $y$ that, when interpreted in $G$ defines the set of
tuples $(l,i,j,v,u)$ such that if there are exactly $k$ walks of
length $l$ starting at $v$ and ending at $u$, then position $[i,j]$ in
the binary expansion of $k$ is 1.  This can be defined by taking the
inductive definition of $\psi^l_k$ we gave in Section~\ref{wlk} and
making the induction part of the formula.

We set out the definition below.
$$
\begin{array}{rcl@{}l}
\psi(\lambda,\kappa_1,\kappa_2,x,y) &  := &
 \ifp_{W,\lambda,\kappa_1,\kappa_2,x,y}[&   \lambda = 1 \land \kappa_1 = 0 \land \kappa_2 =1 \land E(x,y) \lor \\
                                       & & &  \lambda = \lambda'+1 \land \summ(\lambda',\kappa_1,\kappa_2,x,y) ]
\end{array}
$$ where $W$ is a relation variable of type
$(\num,\num,\num,\elem,\elem)$ and the formula $\summ$ expresses that
there is a 1 in the bit position encoded by $(\kappa_1,\kappa_2)$ in
the binary expansion of $k = \sum_{z : E(x,z)} k_{\lambda',z,y}$,
where $k_{\lambda',z,y}$ denotes the number coded by the binary
relation $\{(i,j) : W(\lambda',i,j,z,y)\}$.  We will not write out the
formula $\summ$ in full.  Rather we note that it is easy to define
inductively the sum of a set of numbers given in binary notation, by
defining a sum and carry bit.  In our case, the set of numbers is
given by a ternary relation of type $(\elem,\num,\num)$ where fixing
the first component to a particular value $z$ yields a binary relation
coding a number.  A similar application of induction to sum a set of
numbers then allows us to define the formula
$\phi(\lambda,\kappa_1,\kappa_2)$ which expresses that the bit
position indexed by $(\kappa_1,\kappa_2)$ is 1 in the binary expansion
of $k = \sum_{x \in V} k_x$ where $k_x$ denotes the number coded by
$\{(i,j) : \psi[\lambda,i,j,x,x]\}$.

To define co-spectrality in $\LFPC$ means that we
can write a formula $\cospec$ in a vocabulary with two binary
relations $E$ and $E'$ such that a structure $(V,E,E')$ satisfies this
formula if, and only if, the graphs $(V,E)$ and $(V,E')$ are
co-spectral.  Such a formula is now easily derived from $\phi$.  Let
$\phi'$ be the formula obtained from $\phi$ by replacing all
occurrences of $E$ by $E'$, then we can define:
$$\cospec := \forall \lambda,\kappa_1,\kappa_2 \; \phi \Leftrightarrow \phi'. $$

Now, in order to give a definition in $\PFPC$ of the class of graphs that are DS, we need two variations of the formula $\cospec$.  First, let  $R$ be a relation symbol of type $(\num,\num)$.  We write $\phi(R)$  for the formula obtained from $\phi$ by replacing the symbol $E$ with the relation variable $R$, and suitably replacing number variables with element variables.  So,  $\phi(R,\lambda,\kappa_1,\kappa_2)$ defines, in the graph defined by the relation $R$ on the number domain, the number of closed walks of length $\lambda$.  We write $\cospec_R$ for the formula 
$$\forall \lambda,\kappa_1,\kappa_2 \; \phi(R) \Leftrightarrow \phi, $$
which is a formula with a free relational variable $R$ which, when interpreted in a graph $G$ asserts that the graph defined by $R$ is co-spectral with $G$.  Similarly, we define the formula with two free second-order variables $R$ and $R'$
$$\cospec_{R,R'} := \forall \lambda,\kappa_1,\kappa_2  \;\phi(R) \Leftrightarrow \phi(R'). $$
Clearly, this is true of a pair of relations iff the graphs they define are co-spectral.

Furthermore, it is not difficult to define a formula $\isom(R,R')$ of
$\PFPC$ with two free relation symbols of type $(\num,\num)$ that
asserts that the two graphs defined by $R$ and $R'$ are isomorphic.
Indeed, the number domain is ordered and any property in $\PSpace$ over an ordered domain is definable in $\PFPC$, so such a formula must exist.  Given these, the property of a graph being DS is given by the following formula with second-order quantifiers:
$$\forall R (\cospec_{R} \Rightarrow  \forall R' (\cospec_{R,R'} \Rightarrow \isom(R,R') )).  $$
To convert this into a formula of $\PFPC$, we note that second-order quantification over the number domain can be expressed in $\PFPC$.  That is, if we have a formula $\theta(R)$ of $\PFPC$ in which $R$ is a free second-order variable of type $(\num,\num)$, then we can define a $\PFPC$ formula that is equivalent to $\forall  R \, \theta$.   We do this by means of an induction that loops through all binary relations on the number domain in lexicographical order and stops if for one of them $\theta$ does not hold.  

First, define the formula $\lex(\mu,\nu,\mu',\nu')$ to be the following formula which defines the lexicographical ordering of pairs of numbers:
$$ \lex(\mu,\nu,\mu',\nu') := (\mu < \mu') \lor (\mu = \mu' \land \nu < \nu').$$

We use this to define a formula $\next(R,\mu,\nu)$ which, given a binary relation $R$ of type $(\num,\num)$, defines the set of pairs $(\mu,\nu)$ occurring in the relation that is lexicographically immediately after $R$.
$$
\begin{array}{rcl}
\next(R,\mu,\nu) & := & R(\mu,\nu) \land \exists \mu' \nu' (\lex(\mu',\nu',\mu,\nu) \land \neg R(\mu',\nu')) \lor \\
& & \lor \neg R(\mu,\nu) \land \forall \mu' \nu'  (\lex(\mu',\nu',\mu,\nu) \Rightarrow R(\mu',\nu')) .
\end{array}
$$

We now use this to simulate, in $\PFPC$,  second-order quantification over the number domain. Let $\bR$ be a new relation variable of type $(\num,\num,\num)$ and we define the following formula
$$
\begin{array}{r@{}l}
\forall \alpha \forall \beta \pfp_{\bR,\mu,\nu,\kappa} [ & (\forall \mu \nu \bR(\mu,\nu,0) ) \land \theta(\bR) \land \kappa=0 \lor \\
& \lor \neg\theta(\bR) \land \kappa \neq 0 \lor \\
& \lor \theta(\bR) \land \next(\bR,\mu,\nu) \land \kappa = 0 ] (\alpha,\beta,0) .
\end{array}
$$
It can be checked that this formula is equivalent to $\forall R \, \theta$.

\section{Conclusion}

Co-spectrality is an equivalence relation on graphs with many
interesting facets.  While not every graph is determined upto
isomorphism by its spectrum, it is a long-standing conjecture (see~\cite{Van03}), still open, that \emph{almost all} graphs are DS.
That is to say that the proportion of $n$-vertex graphs that are DS
tends to $1$ as $n$ grows.  We have established a number of results
relating graph spectra to definability in logic and it is instructive
to put them in the perspective of this open question.  It is an easy
consequence of the results in~\cite{Kolaitis92} that the proportion of
graphs that are determined up to isomorphism by their $L^k$ theory
tends to $0$.  On the other hand, it is known that almost all graphs
are determined by their $C^2$ theory (see~\cite{HKL97}) and \emph{a fortiori} by their
$C^3$ theory.  We have established that co-spectrality is
incomparable with $L^k$-equivalence for any $k$; is incomparable with
$C^2$ equivalence; and is subsumed by $C^3$ equivalence.  Thus, our
results are compatible with either answer to the open question of
whether almost all graphs are DS.  It would be interesting to explore
further whether logical definability can cast light on this question.

\bibliographystyle{plain}
\bibliography{main}

\end{document}